\documentclass[11pt]{article}
\usepackage{mathpazo}
\usepackage{amssymb,amsmath,amsthm}
\usepackage{epsfig}
\usepackage{float}
\usepackage{subfigure}

\DeclareMathOperator{\sgn}{sign}

\def\CC{\mathbb{C}}
\def\RR{\mathbb{R}}

\newtheorem{theorem}{Theorem}[section]

\newtheorem{lemma}[theorem]{Lemma}
\newtheorem{proposition}[theorem]{Proposition}

\newtheorem{remark}[theorem]{Remark}

\newtheorem{protocol}[theorem]{Protocol}

\floatstyle{ruled}
\newfloat{protocol}{h}{loa}
\floatname{protocol}{Protocol}
\newfloat{problem}{h}{loa}
\floatname{problem}{Problem}

\newcommand*{\tr}{\mathsf{Tr}}

\DeclareMathOperator{\Ree}{Re}
\DeclareMathOperator{\Imm}{Im}

\newcommand\Exp{{\mathrm E}}
\newcommand\ip[1]{{\langle {#1} \rangle}}
\newcommand{\Vol}{\mathrm{Vol}}
\def\one{\leavevmode\hbox{\small1\normalsize\kern-.33em1}}

\newcommand{\be}{\begin{equation}}
\newcommand{\ee}{\end{equation}}
\newcommand{\bea}{\begin{eqnarray}}
\newcommand{\eea}{\end{eqnarray}}
\newcommand{\bestar}{\begin{equation*}}
\newcommand{\eestar}{\end{equation*}}
\newcommand{\beastar}{\begin{eqnarray*}}
\newcommand{\eeastar}{\end{eqnarray*}}

\newcommand{\eps}{\varepsilon}
\renewcommand{\epsilon}{\varepsilon}

\begin{document}

\title{{\bf Simulating Quantum Correlations with Finite Communication}}
\author{
 Oded Regev\thanks{
School of Computer Science, Tel Aviv University, Tel Aviv 69978, Israel. Supported
   by an Alon Fellowship, by the Binational Science Foundation, by the Israel Science Foundation,
   by the European Commission under the Integrated Project QAP funded by the IST directorate as Contract Number 015848, and by a European Research Council (ERC) Starting Grant.}
 \and
 Ben Toner\thanks{School of Physics, The University of Melbourne,
   Victoria 3010, Australia. Part of this work was completed at CWI
   (Amsterdam) and Caltech.  Supported by NSF Grants PHY-0456720 and CCF-0524828, by ARO Grant
W911NF-05-1-0294, by the Dutch BSIK/BRICKS project, by EU project QAP, and by NWO VICI project 639-023-302.}
}

\maketitle

\begin{abstract}
  Assume Alice and Bob share some bipartite $d$-dimensional quantum
  state. A well-known result in quantum mechanics says that by performing two-outcome measurements,
  Alice and Bob can produce correlations that cannot be obtained
  locally, i.e., with shared randomness alone. We show that by using only two bits of communication,
  Alice and Bob can classically simulate any such correlations. All
  previous protocols for exact simulation required the communication
  to grow to infinity with the dimension $d$.  Our protocol and
  analysis are based on a power series method, resembling Krivine's
  bound on Grothendieck's constant, and on the computation of
  volumes of spherical tetrahedra.
\end{abstract}

\section{Introduction}

\paragraph{Quantum correlations:}
Consider the following game \cite{Clauser:69a}. Alice receives a bit
$a$, and Bob receives a bit $b$, both chosen uniformly at random.
Their task is to output one bit each in such a way that the XOR of the
bits they output is equal to ${\rm AND}(a,b)$. In other words, they should
output the same bit, except when both input bits are $1$.  Notice that
no communication is allowed between them.  A moment's reflection shows
that their best strategy is to always output, say, $0$. This allows
them to win on three of the four possible inputs. It is also not
difficult to show that equipping them with a shared source of
randomness cannot help: the average success probability over the four
possible questions will always be at most $75\%$ (simply because one
can always fix the shared randomness so as to maximize the average
success probability).  This bound of $75\%$, known as the
Clauser-Horne-Shimony-Holt (CHSH) inequality, is the simplest example of a {\em
  Bell inequality} \cite{Bell:64a,Clauser:69a}.

A remarkable and well-known fact is that if Alice and Bob are allowed
to share {\em quantum entanglement} then they can win the game with
probability $\approx 85\%$, no matter which questions are asked.
Indeed, sharing entanglement allows remote parties to realize
correlations that are impossible to obtain classically, without
imparting them with the ability to communicate instantaneously.  This
distinction is one of the most peculiar aspects of quantum theory and
required many years to be properly understood~\cite{Einstein:35a,Bell:64a}.

In this paper we address the topic of quantum correlations from a communication
complexity perspective. Namely, we are asking how many bits of communication
are needed to explain the phenomenon of quantum correlations.
More precisely, we consider the following communication complexity problem,
corresponding to the quantum mechanical scenario of a shared bipartite
quantum state $\rho$ and local two-outcome measurements $\bf A$ and $\bf B$,
with the goal being to simulate the correlation (i.e., the parity)
of the measurement results.

\begin{problem}
\caption{Simulating quantum correlations}\label{prob:simulating}
\begin{tabular}{p{0.18 \textwidth}p{0.75 \textwidth}}
\textbf{Alice's input:}& A $d^2 \times d^2$ positive semidefinite matrix $\rho$ with trace 1 representing
  an operator on the space $\CC^d \otimes \CC^d$
  and a $d\times d$ Hermitian matrix $\bf A$ with $\pm 1$ eigenvalues.\\
\textbf{Bob's input:}& The (same) matrix $\rho$ and a $d\times d$ Hermitian matrix $\bf B$
  with $\pm 1$ eigenvalues.\\
\textbf{Alice's output:}& A bit $\alpha \in \{-1,1\}$. \\
\textbf{Bob's output:}& A bit $\beta \in \{-1,1\}$. \\
\textbf{Goal:}& The correlation $\Exp[\alpha \beta]$ should satisfy
$
\Exp[\alpha \beta] = \tr \left( {\bf A}\otimes{\bf B} \cdot \rho \right).
$
\end{tabular}
\end{problem}
\noindent As we discuss in the open problems paragraph, this problem
is a special case of the problem of simulating local measurements
on quantum states, in which the goal is to simulate the entire output
(as opposed to just the correlation) as well as to handle $m$-outcome
measurements for $m>2$.

In fact, this communication complexity problem can be stated
in an entirely classical and much simpler language which we shall adopt from now on.
The equivalence between the two formulations was established by Tsirelson~\cite{Tsirelson:85b},
and will be described in Appendix~\ref{apndx:classical}.

\begin{problem}
\caption{Simulating quantum correlations (classical formulation)}\label{prob:simulatingclassical}
\begin{tabular}{p{0.18 \textwidth}p{0.7 \textwidth}}
\textbf{Alice:}& Receives as input a unit vector $\vec a \in \RR^n$ and outputs a bit $\alpha \in \{-1,1\}$ \\
\textbf{Bob:}& Receives as input a unit vector $\vec b \in \RR^n$ and outputs a bit $\beta \in \{-1,1\}$ \\
\textbf{Goal:}& The correlation $\Exp[\alpha \beta]$ should satisfy
$
\Exp[\alpha \beta] = \ip{\vec a, \vec b}
$
\end{tabular}
\end{problem}

\noindent
So if $\vec a =\vec b$, Alice and Bob must always output the same
bit, whereas if $\vec a=-\vec b$, they must always output opposite
bits. If, say, $\vec a$ is orthogonal to $\vec b$, then their
outputs should be uncorrelated.

To see how the game described in the beginning of this section fits into this problem,
consider the special case in which Alice's input is either the vector $\vec a_0 = (1,0)$ or
the vector $\vec a_1 = (0,1)$ and Bob's input is either $\vec b_0 = \frac{1}{\sqrt 2}(1,1)$
or $\vec b_1 = \frac{1}{\sqrt 2}(1,-1)$. Notice that $\ip{\vec a_i, \vec b_j}$
is $-\frac{1}{\sqrt{2}}$ if $i=j=1$ and $\frac{1}{\sqrt{2}}$ otherwise.
Therefore, if we are able to simulate quantum correlations in this case,
then we can win the game with probability $\frac{1}{2}+\frac{1}{2\sqrt2}\approx 85\%$.
Using our earlier observations, it follows that even when using shared randomness,
one cannot solve Problem~\ref{prob:simulatingclassical} without any communication, i.e.,
at least one bit of communication is required.

\paragraph{Previous work:}

The problem of simulating quantum correlations was introduced
independently by several authors, including
\cite{Maudlin:92a,Steiner:00a,Brassard:99a}.  It is also closely
related to a communication complexity problem introduced by Kremer,
Nisan, and Ron \cite{KeremerNR99}.  Early work concentrated on the special case of dimension $n=3$, which
turns out to correspond to the case of a shared EPR pair. The protocol
in \cite{Brassard:99a} solves this special case with 8 bits of
communication; this was subsequently improved to just one bit \cite{Toner:03a} (see also \cite{cerf:220403}).
Up to now, the best known protocol for the general case of
Problem~\ref{prob:simulatingclassical} required $\lfloor n/2 \rfloor$
bits of one-way communication~\cite{toner:_how}.

There has also been considerable work on other variants of the question.
For instance, one might consider bounds on the {\em average} communication,
as opposed to the worst case communication as we do here.
The previous best result in this direction is by Degorre, Laplante, and Roland,
who have shown that $(\log n)/2 + O(1)$ bits of communication suffice on average
(over the shared randomness of Alice and Bob), but in their protocol the communication
in the worst case is unbounded~\cite{degorre:062314,degorre:012309}.
Another variant of the question allows for an additive error of at most $\eps$ in the correlations.
In this case there is a straightforward protocol that uses $O(\eps^2\log(1/\eps))$
bits of communication, independent of $n$~\cite{KeremerNR99}.

Our main result improves on all previous work by showing a solution
to Problem \ref{prob:simulatingclassical} using a finite amount of communication,
independent of the dimension $n$.
\begin{theorem}\label{thm:main}
There is a public-coin protocol for exactly simulating quantum correlations
using two bits of one-way communication.
\end{theorem}

We note that the shared randomness is essential: there is no exact private-coin protocol which has bounded communication in the worst case~\cite{PhysRevA.63.052305}.
We also mention that the marginal distributions produced by our protocol (as well as all our
intermediate protocols) are uniform.
This can be verified from the description of the
protocols. Alternatively, note that one can always obtain uniform
marginals without changing the joint correlation by simply taking a
shared random bit $r \in \{-1,1\}$ and asking both Alice and Bob to
multiply their outputs by $r$.

The main question left open in the preliminary version of this
work was whether the theorem is tight, i.e., whether there is a one bit protocol for the problem.
This question has recently been resolved by V{\'e}rtesi and Bene~\cite{VertesiB09}
who showed that no one bit protocol exists.
In Section \ref{sec:applowerbounds} we describe some of our own attempts
to prove such a result. Although our attempts were unsuccessful,
the approaches we describe might be of interest in the future.

\paragraph{Proof outline and techniques:}
We will start in Section~\ref{sec:some-comm-prim} by describing some basic protocols.
All of these protocols have the property that for any input $\vec a$, $\vec b$,
$\Exp[\alpha \beta] = h(\ip{\vec a, \vec b})$ for some
function $h:[-1,1] \to [-1,1]$, i.e., the correlation between
the outputs depends only on the inner product between
the input vectors. (In fact, any protocol can be transformed
into one that has this property: by using the shared randomness,
Alice and Bob can apply a random orthogonal rotation to their inputs;
then, it is not difficult to prove that the resulting distribution on
inputs depends only on the inner product between the original inputs.)

Our goal, of course, is to come up with a protocol whose `correlation
function' $h$ is simply $h(x)=x$. We therefore analyze the correlation
functions of our basic protocols. The main part of the analysis is based on the
calculation of areas of spherical triangles in four-dimensional space
(a topic that was also at the heart of Karloff and Zwick's work
on the approximation of MAX3SAT \cite{KarloffZ97}).
Unfortunately, as we will see in Section~\ref{sec:some-comm-prim}, none of our basic protocols
achieves $h(x)=x$ (see Figures \ref{fig:correlationmaj} and \ref{fig:correlationort}
for plots of some of the correlation functions relative to the
desired $h(x)=x$).

Instead, we will show in Section \ref{sec:simul-joint-corr}
that one can take a protocol whose correlation function $h$
is `strong enough' in some precise sense and transform it into
another protocol whose correlation function is the desired $h(x)=x$.
To complete the proof, we will show in Section~\ref{sec:powerseries}
that the correlation function of our `2-bit orthant protocol'
is strong enough.

The transformation shown in Section \ref{sec:simul-joint-corr} is
the heart of our construction. The idea is to carefully choose
a mapping $C$ from $\RR^n$ to another (infinite-dimensional)
Hilbert space with the property that for any vectors
$\vec a$ and $\vec b$, $\ip{C(\vec a),C(\vec b)} = f(\ip{\vec a, \vec b})$
for some function $f:[-1,1] \to [-1,1]$.
Then, in the transformed protocol, Alice and Bob simply run the
original protocol on inputs $C(\vec a)$ and $C(\vec b)$.
Clearly, this results in a protocol with correlation function
$g(x)=h(f(x))$ where $h$ is the correlation function of the
original function. In order to achieve the desired correlation
function $g(x)=x$ we need to choose $f$ to be $h^{-1}$ (assuming
it is well-defined of course). Intuitively speaking, the purpose of
$C$ is to slightly weaken the correlation function so that it matches
the desired $h(x)=x$.

The main effort, therefore, is in constructing a mapping $C$
with the property that $\ip{C(\vec a),C(\vec b)} = h^{-1}(\ip{\vec a, \vec b})$.
To demonstrate how such a thing can be achieved, assume, for simplicity,
that we have $h^{-1}(x)=x^3$. Then we can choose $C$ to be the mapping
$\vec v \mapsto \vec v \otimes \vec v \otimes \vec v$ where $\otimes$
denotes the tensor operation. It then follows from the definition
that for any vectors $\vec a$ and $\vec b$,
$$\ip{\vec a \otimes \vec a \otimes \vec a, \vec b \otimes \vec b \otimes \vec b} =
   \ip{\vec a,\vec b}^3,$$
as required. In reality, the function $h^{-1}$ will be much more
involved, and we will construct $C$ based on its power series expansion.

The idea of using a power series combined with a mapping $C$ as above
originates in Krivine's work on Grothendieck's
constant~\cite{Krivine:79a}.  More recently, Alon and
Naor~\cite{alon04:_approx} showed that Krivine's method can be
interpreted as an algorithmic rounding technique for a certain family
of semidefinite programs, and this has since been extended in a series
of papers (see, e.g.,
\cite{charikar2004mqp,AroraBHKS05,alon06:_quadr}).
As far as we know, our result is the first application of Krivine's method to
communication complexity.

\paragraph{Open problems:}
The problem we consider in this paper is a special case of the more general
problem of simulating local measurements on quantum states. Here,
as in our problem, Alice and Bob are given (the classical description of)
a bipartite quantum state $\rho$ on $\CC^d \otimes \CC^d$.
In addition, Alice is given an $m$-outcome measurement $A$ and
Bob is given an $m$-outcome measurement $B$.  The goal is for Alice and Bob to output indices
$\alpha, \beta \in [m]$ that are distributed as if they actually
performed the measurements $A$ and $B$ on $\rho$.

The complexity of this general problem is still not well understood, even for very special cases.
Note that we do {\em not} resolve the $m=2$ special case in this paper:
although our protocol gives the correct correlations, it generates uniform marginal distributions,
and not those predicted by quantum theory. The only case that is essentially resolved
is the case $d=m=2$~\cite{CerfGM00,Toner:03a}. Beyond that, no exact protocol with bounded worst-case
communication is known, even for $(d,m) = (2,3)$ or $(d,m) = (3,2)$.
Let us mention some other known results.
First, Brassard, Cleve, and Tapp have established that $\Omega(d)$ bits of communication are necessary for
exact simulation of $d$-outcome measurements on the
maximally-entangled state in $\CC^d \otimes
\CC^d$~\cite{Brassard:99a}. Massar, Bacon,
Cerf, and Cleve have shown that there is an exact private-coin protocol for this problem that
uses $O(d \log d)$ bits of communication on average~\cite{PhysRevA.63.052305}.
On the other hand, the same authors have shown
that any exact protocol with bounded worst-case communication requires
an unbounded number of public coins~\cite{PhysRevA.63.052305}. Finally, Shi and
Zhu have shown that it is possible to simulate the required distribution to
within variational distance $\eps$ using $O({m^6}/{\epsilon^2} \log ({m}/{\epsilon}))$
bits of communication~\cite{Shi05}.

\paragraph{Outline:}
In Section~\ref{sec:some-comm-prim} we present our basic protocols,
and calculate their correlation functions. None of these protocols has the
right correlation function. Then, in Section~\ref{sec:simul-joint-corr} we show
a general technique to take any protocol with a `strong enough' correlation function,
and transform into one that achieves the right correlation function $h(x)=x$.
This part is based on Krivine's power series method.
Finally, in Section~\ref{sec:powerseries} we complete the proof by showing
that one of our basic protocols, the `2-bit orthant protocol',
indeed has a strong enough correlation function.
In Subsection~\ref{ssec:slightlybetter} we show a slightly
better protocol that communicates roughly $1.82$ bits
on average.
There is a discussion of lower bounds in Section~\ref{sec:applowerbounds}.

\section{Basic Protocols}
\label{sec:some-comm-prim}

In this section we present a number of basic communication protocols, and
calculate their correlation functions. None of these protocols has the
right correlation function, but later we will show how to modify them
so that the right correlation function is obtained.

\subsection{Protocol with no communication}\label{ssec:localprim}

The following simple protocol uses no communication and is included for completeness. It is based on the `random hyperplane' idea used in~\cite{Bell:64a,grothendieck53:_resum,GoemansW95}.

\begin{protocol}
\caption{}\label{gl-prot:groth}
\begin{tabular}{p{0.24 \textwidth}p{0.7 \textwidth}}
\textbf{Random Variables:}&
Alice and Bob share a unit vector $\vec \lambda \in \RR^n$ chosen
uniformly at random from the unit sphere.\\
\textbf{Alice:}&
Alice outputs $\alpha = \sgn (\ip{\vec a,\vec \lambda})$.\\
\textbf{Bob:}&
Bob outputs $\beta = \sgn (\ip{\vec b ,\vec \lambda})$.\\
\end{tabular}
\end{protocol}

\begin{lemma}
\label{gl-lemma:grothbasic}
The output of Protocol~\ref{gl-prot:groth} satisfies
\begin{align*}
  \Exp[\alpha \beta] = \frac2\pi \arcsin( \ip{\vec a ,\vec b}).
\end{align*}
In other words, its correlation function is $h(x)=\frac2\pi \arcsin(x)$.
\end{lemma}
\begin{proof}
Let $\vec \mu$ denote the projection of $\vec \lambda$ on the space
spanned by $\vec a$ and $\vec b$, normalized to be of norm $1$. By
symmetry, $\vec \mu$ is distributed uniformly on the unit circle in
that two-dimensional space. Therefore,
$$\Pr(\alpha \neq \beta) = \Pr(\sgn (\ip{\vec \mu, \vec a}) \neq \sgn(\ip{\vec \mu, \vec b}))
    = \frac{1}{\pi} \arccos (\ip{\vec a, \vec b}).$$
It follows that
$$
  \Exp[\alpha \beta] = 1 - 2\Pr(\alpha \neq \beta) =
  \frac{2}{\pi} \arcsin (\ip{\vec a ,\vec b}),
$$
as required.
\end{proof}

\subsection{The majority protocol}
\label{sec:majority-2k+1-bits}

We now present a natural extension of the protocol in the previous
section.  This protocol, which we call the `majority' protocol, will
not be used in the sequel; instead, we will later describe a more
efficient protocol.  We present the majority protocol because its
analysis is somewhat simpler and so it may be useful for generalizing our
technique to simulate stronger correlation functions.

The majority protocol, given as Protocol \ref{basicdp}, is parameterized by a
fixed even integer $k \ge 0$ and uses $k$ bits of one-way communication.
The idea is essentially to repeat the naive random half-space procedure
from the last section $k+1$ times, and then output bits $\alpha$, $\beta$
so that their product is equal to the majority of the $k+1$ products
$\alpha_0 \beta_0,\ldots,\alpha_k \beta_k$ of the outputs of the individual
protocols.
The naive way of implementing this would require sending $k+1$ bits from
Alice and Bob. Namely, Alice outputs $1$ and sends $\alpha_0,\ldots,\alpha_k$
to Bob who outputs $\text{MAJ}(\alpha_0 \beta_0, \alpha_1 \beta_1, \ldots, \alpha_{k} \beta_{k})$.
Instead, a simple trick allows
Protocol \ref{basicdp} to use only $k$ bits: Alice outputs $\alpha_0$ and sends $\alpha_1,\ldots,\alpha_k$
to Bob who outputs $\text{MAJ}(\beta_0, \alpha_0 \alpha_1 \beta_1, \ldots, \alpha_0 \alpha_{k} \beta_{k})$.

As $k$ grows, the correlation function produced by the protocol becomes stronger, as shown
in Figure \ref{fig:correlationmaj}.
It turns out that $k=4$ bits are sufficient to be able to simulate quantum correlations. The proof
of this fact is omitted since we will instead use the more efficient `orthant'
protocol described in Subsection \ref{sec:maximal-k-bit} below.

\begin{figure}[ht]
\center{\epsfxsize=4in\epsfbox{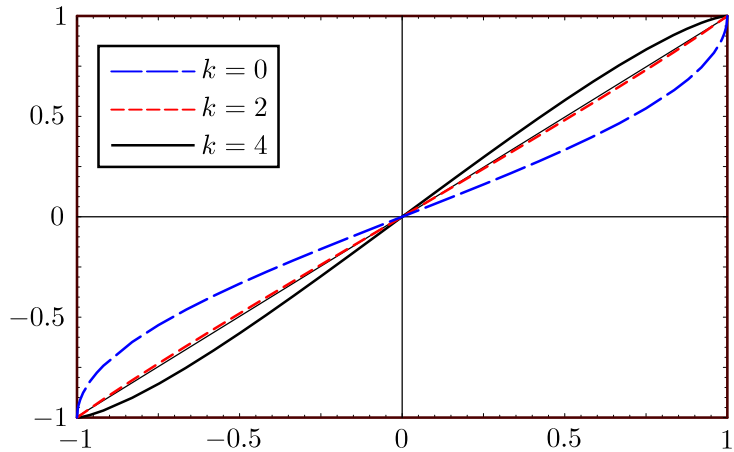}}
 \caption{Correlation functions obtained by the majority protocol relative to the line $h(x)=x$}
 \label{fig:correlationmaj}
\end{figure}

\begin{protocol}
\caption{}\label{basicdp}
\begin{tabular}{p{0.24 \textwidth}p{0.7 \textwidth}}
\textbf{Random Variables:}&
Alice and Bob share $k+1$ unit vectors $\vec
\lambda_i \in \RR^n$ for $i = 0,1, \ldots, k$, chosen independently and
uniformly at random from the unit sphere.\\
\textbf{Alice:}&
Let $\alpha_i = \sgn (\ip{\vec a ,\vec \lambda_i})$
for $i = 0,1,\ldots, k$.  Alice outputs $\alpha=\alpha_0$ and sends to Bob the
$k$ bits $\alpha_0 \alpha_1,\ldots,\alpha_0\alpha_{k}$.\\
\textbf{Bob:}&
Let $\beta_i = \sgn (\ip{\vec b ,\vec \lambda_i})$ for
$i = 0,1, \ldots, k$. Bob outputs
$$
\beta = \text{MAJ}(\beta_0, \alpha_0 \alpha_1 \beta_1, \ldots, \alpha_0 \alpha_{k} \beta_{k}),
$$
where $\text{MAJ}$ is the
majority function.
\end{tabular}
\end{protocol}

Let $g_{k}^{\text{MAJ}}: [0, 1] \to [-1, +1]$ be defined by
$$
g_{k}^{\text{MAJ}}(p) = 1 - 2 \sum_{i=0}^{k/2} \binom{k+1}{i} (1-p)^i p^{k+1-i}.
$$
Let
$$
h_{k}^{\text{MAJ}}(x) = g_{k}^{\text{MAJ}}\left(\frac{1}{\pi} \arccos(x)\right).
$$

\begin{lemma}
\label{aaeaa}
The correlation function of Protocol~\ref{basicdp} is $h_{k}^{\text{MAJ}}$.
\end{lemma}
\begin{proof}
Because the unit vectors $\vec \lambda_i$ are chosen independently,
the events $\alpha_i = \beta_i$ for $i=1,\ldots,k+1$ are independent.
Let $p = \Pr(\alpha_i  \neq \beta_i) = \left({1 - \Exp[\alpha_i \beta_i]}\right)/{2}$, which is independent of $i$.  By Lemma~\ref{gl-lemma:grothbasic},
$$
p = \frac1\pi \arccos(\ip{\vec a ,\vec b}).
$$
Thus
$$
\Pr(\alpha \beta = - 1) = \sum_{i=0}^{k/2} \binom{k+1}{i} (1-p)^i p^{k+1-i}.
$$
Noting that $\Exp[\alpha \beta] = 1 - 2 \Pr(\alpha \beta = - 1)$ completes the proof.
\end{proof}

\subsection{The orthant protocol}
\label{sec:maximal-k-bit}

In this section we present a more efficient protocol that, in some sense,
seems to give the strongest possible correlations. The protocol is
parameterized by an integer $k\ge 0$, and uses $k$ bits of one-way
communication. We call it the `orthant protocol' since it is based on the partitioning
of $k+1$ dimensional space into its $2^{k+1}$ orthants (where an orthant
is the higher-dimensional analogue of the two-dimensional quadrant).
As we shall see below, the correlation function achieved
by this protocol is determined by certain areas on the surface
of the sphere in $k+2$ dimensions. Such questions seem difficult in general
(see \cite{KarloffZ97}). Luckily, for our purposes it suffices to consider
the low-dimensional cases $k=0,1,2$ since the $k=2$ protocol already yields correlations that are
strong enough. The correlation functions produced by the protocol
are shown in Figure \ref{fig:correlationort}.

\begin{figure}[ht]
\center{\epsfxsize=4in\epsfbox{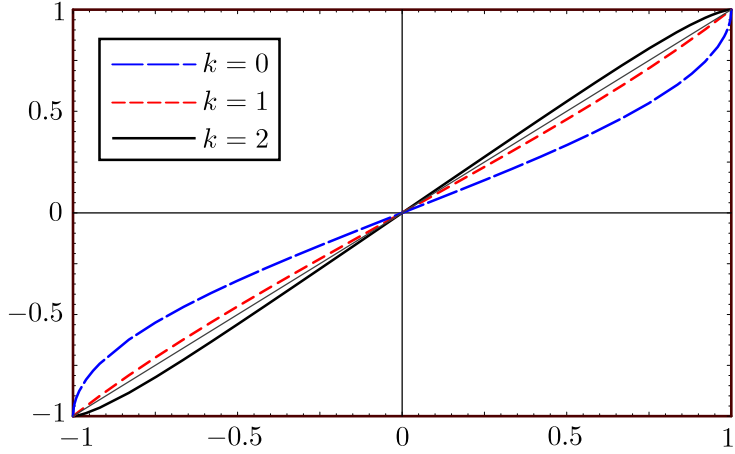}}
 \caption{Correlation functions obtained by the orthant protocol relative to the line $h(x)=x$}
 \label{fig:correlationort}
\end{figure}

\begin{protocol}
\caption{}\label{prot:maximal-k-bit}
\begin{tabular}{p{0.24 \textwidth}p{0.7 \textwidth}}
\textbf{Random Variables:}&
Alice and Bob share a random $(k+1) \times n$
matrix $G$ each of whose entries is an independent standard normal variable, i.e.,
a normal variable with mean 0 and variance 1.\\
\textbf{Alice:}&
Let $\alpha_i = \sgn ( (G \vec a)_i )$
for $i = 0, 1,\ldots, k$ and let $c_i = \alpha_0 \alpha_i$ for $i = 1,
2, \ldots, k$.  Alice outputs $\alpha_0$ and sends to Bob the $k$ bits
$c_1,\ldots,c_{k}$.\\
\textbf{Bob:}&
Bob outputs
$$
  \beta = \sgn \left[ \ip{G\vec b, (1,c_1,\ldots,c_{k})} \right].
$$
\end{tabular}
\end{protocol}

The protocol is given as Protocol \ref{prot:maximal-k-bit}. Roughly speaking,
Alice and Bob start by projecting their vectors onto a random $k+1$-dimensional subspace. Alice then sends to Bob the orthant inside
the $k+1$-dimensional space in which her vector lies, and Bob
uses the half-space determined by this orthant to determine his output.
To be more precise, instead of a random orthogonal projection
we use here a random Gaussian matrix $G$. This leads to a
much cleaner analysis, and moreover, in the limit of large $n$,
the two distributions are essentially the same.
We also use the same trick used in the majority protocol to
reduce the communication from the naive $k+1$ bits to $k$ bits.

We now analyze the correlation function given by this protocol.
For any unit vectors $\vec a, \vec b \in \RR^n$, the output of the protocol satisfies
\begin{align*}
\Exp[\alpha_0 \cdot \beta] &=
   \Exp[\sgn [\alpha_0 \cdot \ip{G\vec b, (1,c_1,\ldots,c_{k})} ]] \\
   &= \Exp[\sgn [\ip{G\vec b, (\alpha_0,\alpha_1,\ldots,\alpha_{k})} ]],
\end{align*}
where expectations are taken over the choice of $G$.
The expression inside the last expectation is $+1$ or $-1$ depending on whether $G\vec b$
is in the half-space defined by the center of the orthant containing
$G\vec a$. By symmetry it is enough to consider the positive orthant, and hence the above is equal to
\begin{align}\label{eq:corrasprob}
 2^{k+2} \Pr\left[\sum_{i=0}^{k} (G\vec b)_i \ge 0 \mbox{~and~} \forall 0 \le i \le k. (G\vec a)_i \ge 0 \right] - 1.
\end{align}

We now claim that the joint distribution of $G \vec a$ and $G \vec b$ is a $2k+2$-dimensional
Gaussian variable with mean $0$ and covariance matrix
\begin{align*}
M = \left(
\begin{array}{cc}
  I & \rho I \\
  \rho I &  I \\
\end{array}
\right),
\end{align*}
where each $I$ is a $(k+1) \times (k+1)$ identity matrix and
$\rho$ denotes the inner product $\ip{\vec a, \vec b}$. To see this, notice
that by the rotational invariance of the Gaussian distribution, we can assume
that $\vec a = (1,0,\ldots,0)$ and $\vec b = (\rho, \sqrt{1-\rho^2},0,\ldots,0)$.
The claim now follows by using the fact that the first two columns of $G$ are
two independent $(k+1)$-dimensional standard Gaussians, i.e., a Gaussian
with mean 0 and covariance $I$.

Our next observation is that the probability in Eq.~\eqref{eq:corrasprob}
depends only on the sum of coordinates of $G\vec b$. We therefore define
the real random variable $Z$ to be $\sum_{i=0}^{k} (G\vec b)_i$.
The joint distribution of $G \vec a$ and $Z$ is given by a $(k+2)$-dimensional Gaussian with mean $0$ and covariance matrix
$$ M' = A M A^t = \left(%
\begin{array}{ccccc}
  1 & 0 & \cdots & 0 & \rho \\
  0 & 1 & \cdots & 0 & \rho \\
  \vdots & \vdots & \ddots &  \vdots& \vdots \\
  0 & 0 & \cdots & 1 & \rho \\
  \rho & \rho & \cdots & \rho & k+1 \\
\end{array}%
\right),$$
where $A$ is the linear transformation taking $(G\vec a, G\vec b)$ to $(G \vec a, Z)$.
We therefore see that the probability in Eq.~\eqref{eq:corrasprob} is exactly
the probability that a vector sampled from a Gaussian distribution with
mean $0$ and covariance matrix $M'$ is in the positive orthant.

By the Cholesky decomposition, we can write $M' = C^t C$ for the $(k+2)\times (k+2)$ matrix
$$ C = \left(%
\begin{array}{ccccc}
  1 & 0 & \cdots & 0 & \rho \\
  0 & 1 & \cdots & 0 & \rho \\
  \vdots & \vdots & \ddots &  \vdots& \vdots \\
  0 & 0 & \cdots & 1 & \rho \\
  0 & 0 & \cdots & 0 & \sqrt{(k+1)(1-\rho^2)} \\
\end{array}%
\right).$$
It is easy to see that
$$ C^{-1} = \left(%
\begin{array}{ccccc}
  1 & 0 & \cdots & 0 & -\frac{\rho}{\sqrt{(k+1)(1-\rho^2)}} \\
  0 & 1 & \cdots & 0 & -\frac{\rho}{\sqrt{(k+1)(1-\rho^2)}} \\
  \vdots & \vdots & \ddots &  \vdots& \vdots \\
  0 & 0 & \cdots & 1 & -\frac{\rho}{\sqrt{(k+1)(1-\rho^2)}} \\
  0 & 0 & \cdots & 0 & \frac{1}{\sqrt{(k+1)(1-\rho^2)}} \\
\end{array}%
\right).$$
Since $(C^{-1})^t M' C^{-1} = I$, applying the linear transformation $(C^{-1})^t$
to a Gaussian random variable with mean $0$ and covariance matrix $M'$ transforms
it into a standard Gaussian variable. Under this transformation,
the positive orthant, which is the cone spanned by the standard
basis vectors, becomes the cone spanned by the rows of $C^{-1}$.
We conclude that the probability in Eq.~\eqref{eq:corrasprob} is exactly
the probability that a vector sampled from a standard Gaussian distribution
is in the cone spanned by the rows of $C^{-1}$. By the spherical symmetry
of the standard Gaussian distribution, we can equivalently ask
for the relative area of the sphere $S^{k+1} \subset \RR^{k+2}$
that is contained inside the cone spanned by the rows of $C^{-1}$.

\paragraph{The case $k=0$:}
We can now compute the probability in Eq.~\eqref{eq:corrasprob} for each of $k=0,1,2$.
We start with the simplest case of $k=0$. Here, we are interested
in the relative length of the circle $S^1$ contained in the
cone spanned by the rows of $C^{-1}$. Obviously, this is given
by the angle between the two vectors divided by $2\pi$, which is
$\arccos(-\rho)/(2\pi)$.
Hence by Eq.~\eqref{eq:corrasprob} the correlation function in this case is simply
$$ h_0^{\text{ORT}}(\rho) := \frac{2}{\pi}\arccos(-\rho) - 1 = \frac{2}{\pi}\arcsin(\rho).$$
We could also obtain this result by noting that the $k=0$ protocol
is essentially identical to the one from Section~\ref{ssec:localprim}.

\begin{figure}[ht]
\center{\epsfxsize=2.5in\epsfbox{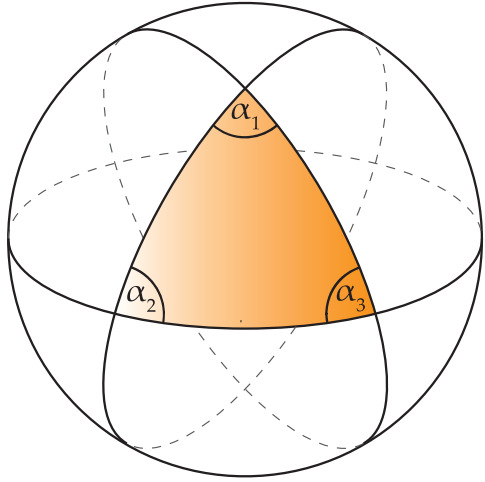}}
 \caption{A spherical triangle}
 \label{fig:sphericaltriangle}
\end{figure}

\paragraph{The case $k=1$:}
We now analyze the more interesting case $k=1$.
Here, we are interested in the relative area of the sphere $S^2$
contained in the cone spanned by the three rows of $C^{-1}$.
The intersection of $S^2$ with a cone spanned by three vectors
is known as a {\em spherical triangle}, see Figure \ref{fig:sphericaltriangle}.
Its area, as given by Girard's formula (see, e.g., \cite[Page 278]{BergerBook}), is
$\alpha_1 + \alpha_2 + \alpha_3 - \pi$ where $\alpha_1,\alpha_2,\alpha_3$
are the three angles of the triangle (as measured on the surface).
In more detail, if $v_1,v_2,v_3$ are the vectors spanning the
cone, then $\alpha_1$ is the angle between the
two vectors obtained by projecting $v_2$ and $v_3$ on the plane
orthogonal to $v_1$ (and similarly for $\alpha_2$ and $\alpha_3$).
In our case, the cone is spanned by $v_1=(\sqrt{2(1-\rho^2)},0,-\rho)$,
$v_2=(0,\sqrt{2(1-\rho^2)},-\rho)$, and $v_3=(0,0,1)$. Clearly
$\alpha_3 = \pi/2$ and a short calculation shows that
$\alpha_1=\alpha_2=\arccos(-\rho/\sqrt{2})$. Plugging this
into Girard's formula, and using the fact that the area
of the sphere is $4\pi$, we obtain that the relative area
of $S^2$ contained in the cone spanned by the rows of $C^{-1}$
is $(2 \arccos(-\rho/\sqrt{2}) - \pi/2)/(4\pi)$.
Hence by Eq.~\eqref{eq:corrasprob} the correlation function in this case is
\begin{align*}
h_1^\text{ORT}(\rho) := \frac{4}{\pi} \arccos(-\rho/\sqrt{2}) - 2 = \frac{4}{\pi} \arcsin(\rho/\sqrt{2}).
\end{align*}

\paragraph{The case $k=2$:}
We finally arrive at the most important case $k=2$. Here we are considering
{\em spherical tetrahedra}, defined as the intersection of $S^3$ with a
cone spanned by four vectors. Unlike the case of spherical triangles,
no closed formula is known for the volume of a spherical tetrahedron
(see \cite{KarloffZ97} for further discussion and references). Fortunately,
there is a simple formula for the {\em derivative} of the volume, as
we describe in the sequel.

We start with some preliminaries on spherical tetrahedra, closely following
Appendix A in \cite{KarloffZ97}. A spherical
tetrahedron is defined by four unit vectors $v_0,v_1,v_2,v_3 \in S^3$
forming its vertices. For $0\le i<j \le 3$ let $\theta_{ij} = \arccos(\ip{v_i,v_j})$
be the angle between $v_i$ and $v_j$. Equivalently, $\theta_{ij}$ is
the spherical length of the edge $ij$. Another set of six parameters
associated with a spherical tetrahedron are its {\em dihedral angles}
$\lambda_{ij}$, $0\le i<j \le 3$, describing the angle between the two
faces meeting at the edge $ij$. They are defined as
$$
\lambda_{01} = \arccos \frac{\ip{v_0 \wedge v_1 \wedge v_2, v_0 \wedge v_1 \wedge v_3}}
  {|v_0 \wedge v_1 \wedge v_2| |v_0 \wedge v_1 \wedge v_3|}
$$
and similarly for the other five dihedral angles, where the {\em high-dimensional
inner product} is defined as
$$
\ip{a_1 \wedge a_2 \wedge a_3, b_1 \wedge b_2 \wedge b_3} = \det \left(%
\begin{array}{ccc}
  \ip{a_1,b_1} & \ip{a_1,b_2} & \ip{a_1,b_3} \\
  \ip{a_2,b_1} & \ip{a_2,b_2} & \ip{a_2,b_3} \\
  \ip{a_3,b_1} & \ip{a_3,b_2} & \ip{a_3,b_3}\\
\end{array}%
\right)
$$
and the {\em high-dimensional norm} is given by
$$
|a_1 \wedge a_2 \wedge a_3| = \ip{a_1 \wedge a_2 \wedge a_3, a_1 \wedge a_2 \wedge a_3}^{1/2}.
$$
Finally, in order to compute the volume of a spherical tetrahedron we
use a formula due to Schl\"afli \cite{Schlafli1858}, which says that for every $0 \le i<j\le 3$,
$$
\frac{\partial \Vol}{\partial \lambda_{ij}} = \frac{\theta_{ij}}{2},
$$
where $\Vol = \Vol(\lambda_{01},\lambda_{02},\lambda_{03},\lambda_{12},\lambda_{13},\lambda_{23})$ is the volume of a spherical tetrahedron with
the given edge lengths.

Our goal is to compute the volume of the spherical tetrahedron whose
vertices are the rows of $C^{-1}$ normalized to be of norm $1$,
\begin{align*}
v_0 &= \left(\sqrt{3-3\rho^2},0,0,-\rho \right) / \sqrt{3-2\rho^2}\\
v_1 &= \left(0,\sqrt{3-3\rho^2},0,-\rho \right) / \sqrt{3-2\rho^2} \\
v_2 &= \left(0,0,\sqrt{3-3\rho^2},-\rho \right) / \sqrt{3-2\rho^2}\\
v_3 &= \left(0,0,0,1\right).
\end{align*}
 From this it easily follows that
$$ \theta_{03} = \theta_{13} = \theta_{23} = \arccos(-\rho / \sqrt{3-2\rho^2})
  \quad \mbox{and} \quad
  \theta_{01} = \theta_{02} = \theta_{12} = \arccos(\rho^2 / (3-2\rho^2)).$$
Moreover, a straightforward calculation reveals that
$$ \lambda_{03} = \lambda_{13} = \lambda_{23} = \pi/2
  \quad \mbox{and} \quad
  \lambda_{01} = \lambda_{02} = \lambda_{12} = \arccos(-\rho / \sqrt{3}),$$
and that the derivative of the latter term as a function of $\rho$
is $(3-\rho^2)^{-1/2}$.
By using Schl\"afli's formula and integrating along $\rho$, we obtain that the volume of our spherical tetrahedron is
$$ \int_{-1}^{\rho} 3 \cdot \frac{1}{2} \arccos(\sigma^2 / (3-2\sigma^2)) \cdot (3-\sigma^2)^{-1/2} d\sigma, $$
where we used that for $\rho=-1$ this volume is $0$. Since the total area of
$S^3$ is $2 \pi^2$, we obtain using Eq.~\eqref{eq:corrasprob}
that the correlation function in this case is
\begin{align}\label{eq:correlationk2}
h_2^\text{ORT}(\rho) := \frac{12}{\pi^2} \int_{-1}^{\rho} \frac{\arccos(\sigma^2 / (3-2\sigma^2))}{\sqrt{3-\sigma^2}} d\sigma - 1.
\end{align}

\section{Simulation of the Joint Correlation}
\label{sec:simul-joint-corr}

In this section we describe how to take any protocol
whose correlation function is `strong enough', and use it
to solve Problem~\ref{prob:simulatingclassical}. This section as
well as the next one rely on some basic facts from the theory
of real analytic functions which can be found in, e.g., \cite{RealAnalyticBook}.
As we said earlier, the idea is to carefully choose a mapping
$C$ such that when Alice and Bob apply the protocol on $C(\vec a)$
and $C(\vec b)$, the resulting correlation function will be correct.

\begin{protocol}
\caption{Transformed protocol}
\label{two-prot:infinite}
Alice and Bob map their vectors to $C(\vec a)$ and $C(\vec b)$ and run
the original protocol on these vectors.
\end{protocol}

Fix some arbitrary correlation function $h:[-1,1] \to [-1,1]$.
We now give sufficient conditions on $h$ under which the required
transformation $C$ exists.
First, we require that $h(1)=1$ and that $h$ is odd, continuous, and monotonically increasing.
Moreover, we require that its series expansion about $0$,
\begin{align}\label{eq:seriesofh}
h(x) = \sum_{k=0}^\infty {c_{2k+1}} x^{2k+1},
\end{align}
converges to $h(x)$ on the interval $(-1,1)$, which implies that $h$ is (real) analytic on $(-1,1)$.
Finally, we require that $c_1>0$ and $c_{2k+1} \leq 0$ for all $k >0$.

In Section~\ref{sec:powerseries} we will show that the orthant protocol with $k=2$
satisfies these properties.
A crucial fact for our protocol is that under the above requirements on $h$,
the power series of $h^{-1}$ converges on $[-1,1]$ and all its coefficients are nonnegative.
This is shown in the following lemma.

\begin{lemma}\label{lemma:1}
If $h$ satisfies the above conditions, then $h^{-1}$ has a power series expansion
\begin{align}\label{eq:seriesofhinv}
h^{-1}(x) = \sum_{k=0}^\infty {d_{2k+1}} x^{2k+1}
\end{align}
that converges on $[-1,1]$ and satisfies $d_{2k+1} \ge 0$ for all $k \ge 0$.
\end{lemma}

\begin{proof}
We first notice that under the above conditions, the power
series in Eq.~\eqref{eq:seriesofh} converges to $h$ also at the endpoints $-1,1$.
This is easy to prove and follows from Abel's lemma,
which says that if all but finitely many of the coefficients
of a power series are nonnegative (or nonpositive), then the value of
the series at the endpoints is given by the limit of its values
as we approach the endpoint. We hence see that
$\sum_{k=0}^\infty {c_{2k+1}} = 1$.

It also follows easily that the inverse function $h^{-1}:[-1,1] \to [-1,1]$
is well-defined and is odd. Moreover,
by the real analytic inverse function theorem (see \cite[Theorem 1.5.3]{RealAnalyticBook}),
$h^{-1}$ is analytic on $(-1,1)$ and hence has a series expansion about $x=0$,
as in Eq.~\eqref{eq:seriesofhinv}. In order to analyze this series,
we use a known formula for the coefficients of an inverse function (see, e.g.,
\cite[Eq.~(4.5.12)]{MS}):
$$
d_{k} = \frac{1}{k c_1^{k}} \sum_{\ell_1, \ell_2, \ldots}
 \frac{(k) (k+1) \cdots (k - 1 +
\ell_1 + \ell_2  + \cdots )}{\ell_1!\ell_2!\ell_3!\cdots}
\left(-\frac{c_2}{c_1}\right)^{\ell_1}
\left(-\frac{c_3}{c_1}\right)^{\ell_2}\cdots,
$$
where the sum runs over nonnegative integers satisfying $\ell_1 +
2\ell_2 + 3\ell_3 + \cdots = k-1$.  Since in our case every term in the sum is
nonnegative, it follows that $d_k \geq 0$ for all $k$, as required.

It remains to show that the series converges on $[-1,1]$.
In fact, it is sufficient to show that the series converges on $(-1,1)$:
convergence at the endpoints $-1,1$ would follow by Abel's lemma, as before.
We do this by showing that $d_{k} \leq 1/k$, as this immediately implies that
the power series converges on the interval $(-1,1)$.
Using the above formula, we get that for all $k >0$
\begin{align*}
d_{2k+1} &= \frac{1}{(2k+1) c_1^{2k+1}} \sum_{\ell_2, \ell_4,
\ldots} \frac{(2k+1) (2k+2) \cdots (2k + \ell_2 + \ell_4 + \cdots
)}{\ell_2!\ell_4!\cdots} \left(-\frac{c_3}{c_1}\right)^{\ell_2}
\left(-\frac{c_5}{c_1}\right)^{\ell_4}\cdots,
  \end{align*}
where the sum runs over nonnegative integers satisfying $\ell_2 + 2
\ell_4 + 3 \ell_6 + \cdots = k$.  We extend the sum to all
nonnegative integers $\ell_2, \ell_4, \ldots$, obtaining
\begin{align}
d_{2k+1}
 &\leq \frac{1}{(2k+1) c_1^{2k+1}} \sum_{m=0}^\infty \binom{2k+m}{m}\sum_{\ell_2, \ell_4, \ldots} \frac{m!}{\ell_2!\ell_4!\cdots} \left(-\frac{c_3}{c_1}\right)^{\ell_2} \left(-\frac{c_5}{c_1}\right)^{\ell_4}\cdots \label{eq:1}\\
 & = \frac{1}{(2k+1) c_1^{2k+1}} \sum_{m=0}^\infty
\binom{2k+m}{m}\left(\frac{- c_3 - c_5 - \cdots}{c_1}\right)^m,
\label{eq:2}
\end{align}
where the inner sum in Eq.~\eqref{eq:1} is over all indices $\ell_2
+ \ell_4 + \cdots =m$ and we used the multinomial theorem to obtain
Eq.~\eqref{eq:2}.  By our observation above,
$- \sum_{k=1}^\infty c_{2k+1} = c_1 - 1$.
Set $z = 1-1/c_1$.  Then $0\leq z <1$ and
\begin{align*}
d_{2k+1} \leq \frac{1}{2k+1} (1-z)^{2k+1} \sum_{m=0}^\infty
\binom{2k+m}{m} z^m  = \frac{1}{2k+1} (1-z)^{2k+1}(1-z)^{-(2k+1)} =
\frac{1}{2k+1},
  \end{align*}
  since the sum is just the negative binomial series.  We conclude
  that the power series converges to $h^{-1}$ on the interval
  $(-1,1)$, which also implies convergence at the endpoints by
  Abel's lemma.
\end{proof}

The transformation $C$ is obtained by applying the following lemma to $h^{-1}$.

\begin{lemma}\label{lemma:transformationC}
Let $f:[-1,1] \to [-1,1]$ be a function with a power series expansion
$f(x) = \sum_{k=0}^\infty {d_{k}} x^{k}$ that converges on $[-1,1]$ and satisfies $d_{k} \ge 0$ for all $k \ge 0$
and $f(1)=1$. Then for any $n\ge 1$, there exists a transformation $C:S^{n-1} \to S^\infty$ such
that for all $\vec a, \vec b \in S^{n-1}$, $\ip{C(\vec a) ,C(\vec b)} =
f \big(\ip{\vec a ,\vec b} \big).$
\end{lemma}
\begin{proof}
Define
$$
   C(\vec v) = \bigoplus_{k=0}^\infty \sqrt{{d_{k}}}\, {\vec v}\,^{\otimes k},
$$
where ${\vec v}\,^{\otimes k}$ denotes the vector $\vec v \otimes \vec v \otimes
\cdots \otimes \vec v$ with $k$ tensor factors.
Note that this is well-defined, since $d_{k} \geq 0$ for all $k$.
By definition, for any $\vec a, \vec b \in S^{n-1}$ we have
$$
 \ip{C(\vec a) ,C(\vec b)} =
 \sum_{k=0}^\infty d_{k} \ip{\vec a^{\otimes k} ,\vec b^{\otimes k}} =
 \sum_{k=0}^\infty d_{k} \ip{\vec a ,\vec b }^{k} =
 f \big(\ip{\vec a ,\vec b} \big),
$$
which in particular implies that $C(\vec a)$ is a unit vector for any $\vec a \in S^{n-1}$.
\end{proof}

\begin{remark}
\label{remark:schoenberg}
By Schoenberg's theorem~\cite[Theorem 2]{MR0005922}, the conditions on
$f$ in Lemma~\ref{lemma:transformationC} are in fact necessary for the transformation
$C$ to exist.
It is also known that if we are only interested in a transformation $C$ for
a particular value of $n$ (rather than for all $n \ge 1$),
it is sufficient (and necessary) to require that $f$, in addition to satisfying $f(1)=1$,
has a non-negative convergent series expansion in Gegenbauer polynomials~\cite{Krivine:79a,MR0005922}.
\end{remark}

\begin{theorem}
Protocol~\ref{two-prot:infinite} is well-defined and solves Problem~\ref{prob:simulatingclassical}.
\end{theorem}
\begin{proof}
By Lemma~\ref{lemma:1}, we can apply Lemma~\ref{lemma:transformationC} to $h^{-1}$
to obtain the transformation $C$. Since $C$ maps unit vectors to unit vectors,
Protocol~\ref{two-prot:infinite} is well-defined.
Moreover, its output satisfies
$$
\Exp[\alpha \beta] = h\big( \ip{C(\vec a) ,C(\vec b)} \big) =
  h\circ h^{-1} \big( \ip{\vec a ,\vec b} \big) =
  \ip{\vec a ,\vec b},
$$
as required.
\end{proof}

\begin{remark}
  As mentioned in Remark~\ref{remark:schoenberg}, the properties
  of $h^{-1}$ in the conclusion of Lemma~\ref{lemma:1} are not just
  sufficient, but also necessary for Protocol~\ref{two-prot:infinite} to be well-defined.
  It is therefore natural to ask if the conditions on $h$
  in Lemma~\ref{lemma:1} are also necessary. It turns
  out that they are not: take, for example, $h$ such that $h^{-1}(x) = 0.9 x +
  0.1 x^3$, calculate that $h(x) = 1.11 x - 0.15 x^3 + 0.06 x^5 +
  O(x^7)$, and notice that the coefficient of $x^5$ is positive.
  So is there a necessary and sufficient condition?
  We do not know. Looking at Figure~\ref{prot:maximal-k-bit},
  one might be tempted to replace
  the condition that $c_1>0$ and $c_{2k+1} \leq 0$ for all $k >0$ with
  the weaker condition that $h(x) \geq x$ for $x \geq 0$.
  But this condition is not sufficient: the function $h(x) = x + 0.1 x^3 - 0.1 x^5$
  satisfies $h(x)\geq x$ for $x \ge 0$ (as well as our other requirements), but
  $h^{-1}(x) = x - 0.1 x^3 + O(x^5)$.
\end{remark}

\def\ax{\arccos x}
\def\bax{\left( \arccos x \right)}
\def\as{\arcsin x}
\def\bas{\left( \arcsin x \right)}

\section{Analysis of the Power Series}
\label{sec:powerseries}

In this section, we show that the orthant protocol with $k=2$ satisfies the requirements listed in
Section~\ref{sec:simul-joint-corr} and hence can be used
to simulate quantum correlations with only two bits of communication.
A similar but much more involved analysis holds also for the majority
protocol with $k=4$ and implies a protocol with four bits of communication.
We omit this analysis since the orthant protocol is superior in all respects.

\begin{lemma}
Let $h(x) = h_{2}^\text{ORT}(x)$ be as given in Eq.~\eqref{eq:correlationk2}.
Then $h(1)=1$ and $h$ is odd, continuous, and monotonically increasing.
Moreover, it is (real) analytic on $(-1,1)$, and its power series about $x=0$,
$$
h(x) = \sum_{k=0}^\infty {c_{2k+1}} x^{2k+1},
$$
satisfies $c_1>0$ and $c_{2k+1} < 0$ for all $k >0$.
\end{lemma}

\begin{proof}

The first four conditions are obvious. Moreover, being composed of
analytic functions, $h$ is easily seen to be analytic on $(-1,1)$.
We now show that the coefficients $c_{2k+1}$ have the right sign.
 From the derivative
$$
h'(x) = \frac{12 \arccos(x^2/(3-2x^2))}{\pi^2 \sqrt{3-x^2}}
$$
it follows that $c_1 = h'(0) = 2 \sqrt3 /\pi > 0$.
To show that the rest of the coefficients are negative, consider the second derivative,
$$
h''(x) = -\frac{24 x}{\pi^2}H_1(x^2) H_2(x^2)
$$
where
\begin{align}
  \label{eq:h1ink2}
  H_1(t) &= \frac{1}{(3 - t)^{3/2} (3-2t)},\\
H_2(t) &= \sqrt{\frac{3(3-t)}{1-t}} - \frac{3-2t}{2} \arccos \left( \frac{t}{3-2t}\right). \nonumber
\end{align}
It is clear from Eq.~\eqref{eq:h1ink2} that all the coefficients in the power
series of $H_1$ about $0$ are positive.
Therefore, it is sufficient to show that all the coefficients in the power series of $H_2$ are positive.
We calculate $H_2(0) = 3 - 3\pi/4>0$ and $H_2'(0) = (3+\pi)/2 >0$.  Next, we calculate
$$
H_2''(t)  = \frac{\sqrt3}{2} \left( 7 + \frac{3}{1-t} \right) \frac{1}{(3-2t)(1-t)^{3/2}\sqrt{3-t}}.
$$
Hence all coefficients in the power series of $H_2$ are positive, as required.
\end{proof}

\subsection{The 1.82 bit protocol}
\label{ssec:slightlybetter}

In this section we observe that the amount of communication can be
lowered to $1.82$ bits on average by performing the $k=1$ orthant
protocol with probability
$$
  p := \frac{8-2 \pi }{8+\left(\sqrt{6}-2\right) \pi } \approx 0.18
$$
and the $k=2$ orthant protocol the remainder of the time.  This is the largest
value of $p$ for which our protocol works.  One could possibly obtain
a slightly better protocol by directly modifying the $k=2$ orthant protocol,
but the analysis seems to get too complicated.

\begin{lemma}
Let $h(x) = p h_1^\text{ORT}(x) + (1-p) h_2^\text{ORT}(x)$.
Then $h(1)=1$ and $h$ is odd, continuous, and monotonically increasing.
Moreover, it is (real) analytic on $(-1,1)$, and its power series about $x=0$,
$$
h(x) = \sum_{k=0}^\infty {c_{2k+1}} x^{2k+1},
$$
satisfies that $c_1>0$ and $c_{2k+1} \le 0$ for all $k >0$.
\end{lemma}

\begin{proof}
The first five properties are easy to verify as before.
A short calculation shows that $c_1 > 0$.
Then calculate
$$
  h''(x) = -\frac{24 x}{\pi^2} H_1(x^2) H_2(x^2)
$$
where
\begin{align*}
  H_1(t) &= \frac1{(2-t)^{3/2}}\\
H_2(t) &= -p \frac{\pi}{6} + (1-p)\frac{1}{(3-t)^{3/2}}H_3(t)\\
H_3(t) &= (2-t)^{3/2} \left[\frac{1}{3-2t}\sqrt{\frac{{3(3-t)}}{{1-t}}} - \frac{1}{2} \arccos \left( \frac{t}{3-2t}\right)\right].
\end{align*}
It is clear that the coefficients $c_{2k+1}$ will be nonpositive for all $k > 0$ if $H_1(t)$
and $H_2(t)$ have series expansions with nonnegative coefficients.  This is clear for $H_1(t)$.
For $H_2(t)$, first notice that $H_2(0) = 0$
(this explains our choice of $p$) and
then note that it has a series expansion with nonnegative coefficients if $H_3(t)$ does.
We calculate $H_3(0) = (4-\pi)/\sqrt 2>0$, and differentiate to obtain
$$ H_3'(t) = \frac{{{3 \sqrt{ 2 - t }}}}{4} \left[ \frac{{{2 \sqrt{ 3 - t }\left( 5 - {4 t} \right)} }}{{{\sqrt{
3 } {\left( 3 - {2 t} \right)}^{2} } {\left( 1 - t
\right)}^{{3}/{2}} }}+  \arccos \left( \frac{t}{3 - {2 t}}
\right) \right].
$$
 From this we calculate $H_3'(0) = (20+9\pi)/12\sqrt2 >0$ and then
differentiate again, finding
$$
  H_3''(t) = \frac{3\sqrt{3}}{4\sqrt{2-t}} H_4(t),
$$
where
$$
  H_4(t) = \frac{79 -157 t+85 t^2+11 t^3-20 t^4+4 t^5}{(3-2 t)^3 (1-t)^{5/2} \sqrt{3-t}} -
     \frac{1}{2\sqrt{3}} \arccos \left( \frac{t}{3-2t} \right).
$$
We calculate $H_4(0) = (316 - 27\pi)/108\sqrt{3}>0$, differentiate once more and find that
\begin{align*}
H_4'(t) &=
\frac{1}{54 \sqrt{1-t} (3-t)^{3/2}} \left[
  1297
+ \frac{9399  t} {(3-2 t)}
+ \frac{30696 t^2} {(3-2 t)^2}
+ \frac{66116  t^3} {(3-2 t)^3}
+ \frac{115080 t^4} {(3-2 t)^4} \right.\\
&\hspace{1cm}\left.
+ \frac{59616 t^5} {(3-2 t)^4(1-t)}
+ \frac{8748 t^6} {(3-2 t)^4(1-t)^2}
+ \frac{540 t^7} {(3-2 t)^4(1-t)^3}\right],
\end{align*}
which we have written in a form that makes it clear that $H_4(t)$ has a series expansion with nonnegative coefficients.  Tracing backwards through the proof, we conclude that the coefficients in the series expansion of $h(x)$ about $x = 0$ have the desired property.
\end{proof}

\section{Lower Bounds}\label{sec:applowerbounds}

As mentioned in the introduction, it has recently been shown
by V{\'e}rtesi and Bene~\cite{VertesiB09} that our main theorem
is tight, i.e., no one bit protocol exists for Problem~\ref{prob:simulatingclassical}.
Here we describe some of our own attempts to prove this, which, although
unsuccessful, might shed further light on the complexity of the problem.

Consider the special case when $\vec a$ and $\vec b$ are both uniformly distributed in $S^{n-1}$ but
constrained so that $\ip{\vec a, \vec b} = \pm (1-\eps)$ with $0< \eps
\ll 1$. Given a one bit protocol $\cal P$, define the function
\begin{align*}
  B_n(\epsilon) = 2-\Exp_{\ip{\vec a, \vec b} = 1-\eps} \left[\Exp_{\cal P}[\alpha \beta]\right] + \Exp_{\ip{\vec a, \vec b} = -1+\eps}\left[\Exp_{\cal P}[\alpha \beta]\right],
\end{align*}
where the outer expectations are taken over all vectors $\vec a, \vec
b \in S^{n-1}$ with inner product $\pm (1-\eps)$ and the inner
expectations are taken over any shared randomness used by the protocol
$\cal P$. If $\cal P$ solves Problem~\ref{prob:simulatingclassical}, then we
necessarily have $\Exp_{\cal P}[\alpha \beta] =
\ip{\vec a, \vec b}$, and hence $B_n(\eps) = 2\eps$ for all $n$ and $\eps$.
Our earlier analysis shows that for the orthant protocol (Protocol~\ref{prot:maximal-k-bit})
with $k=1$, as $\eps$ goes to $0$, $B_n(\eps)$ approaches $8\eps/\pi > 2\eps$.
We conjecture that for sufficiently large $n$ and sufficiently small $\eps$,
all one-bit protocols satisfy $B_n(\eps) > 2\eps$ (and therefore do not
solve Problem~\ref{prob:simulatingclassical}). In fact, we conjecture
that the orthant protocol with $k=1$ is optimal for $B_n(\eps)$, i.e.,
we conjecture that for all one-bit protocols, $\lim_{n\to \infty, \eps \to 0} B_n(\eps)/\eps \geq
8/\pi$. 

One approach to prove these conjectures is the following. First, since we are only interested in
minimizing the value of $B_n(\eps)$ and not in obtaining the correct
correlations, we can restrict attention to deterministic protocols.
Any deterministic protocol partitions Alice's sphere $S^{n-1}$ into
four sets $R(\alpha, c)$, depending on which bit $\alpha$ she outputs
and which bit $c$ she sends.  Once we have specified Alice's strategy,
we can assume Bob acts optimally to minimize $B_n(\eps)$.  The
contribution to $B_n(\eps)$ comes from regions near which $R(-1, c)$
meets $R(+1,c)$, since it is in these areas that Bob cannot tell whether Alice outputs $+1$ or $-1$,
and hence cannot correlate his answer perfectly with Alice's.
Therefore, in order to prove the conjectures, one should argue that any
protocol $\cal P$ must have local regions where $R(-1,c)$ meets
$R(+1,c)$ and that the way these regions meet in the $k=1$ orthant
protocol is optimal.  Formalizing this notion would seem to require
topological arguments, perhaps an extension of the Borsuk-Ulam
theorem.

\medskip
In another attempt to shed light on the problem,
we show now how to extend a lower bound of Barrett, Kent, and
Pironio~\cite{BarrettKP06}, who improved on an earlier
result of Pironio~\cite{PhysRevA.68.062102}. Barrett et al.\ showed
that if we examine the transcript of communication between Alice and
Bob of any protocol for Problem~\ref{prob:simulatingclassical}, then
with probability $1$ (over the shared randomness used by the protocol)
the transcript must show some communication. In other words, it cannot
be the case that Alice and Bob sometimes output results using shared
randomness alone. But this leaves open the possibility that, say,
Alice almost always sends the same message to Bob.

Here, we show a lower bound on the (min-)entropy of the communication
transcript.  More specifically, we show an upper bound on the maximum
probability with which a transcript can appear in a protocol for
Problem~\ref{prob:simulatingclassical}.

\begin{proposition}
There exists a distribution on inputs such that in any protocol that
solves Problem~\ref{prob:simulatingclassical} no transcript can appear
with probability greater than $(3-\sqrt{2})/2 \approx 0.79$ when applied to this input distribution.
\end{proposition}

\begin{proof}
Let $P$ be a protocol that solves Problem~\ref{prob:simulatingclassical}.
As mentioned in the introduction, $P$ in particular
allows us to solve the following problem with probability
$\frac{1}{2}+\frac{1}{2\sqrt{2}}\approx 0.85$: Alice and Bob
receive bits $a$ and $b$ respectively, and their task is to output
one bit each in such a way that the XOR of the bits they output
is equal to ${\rm AND}(a,b)$. Assume the bits $a$ and $b$ are chosen uniformly
at random, and consider the resulting distribution on transcripts
created by $P$. Consider the most likely transcript $T$
and let $p$ denote the probability with which it occurs. We now
construct a protocol $P'$ with no communication as follows.
Alice checks whether the transcript $T$ is consistent with her input.
If so, she outputs a bit as in $P$; if not, she outputs a random bit.
Bob does the same. Note that with probability $p$, $P'$ behaves
identically to $P$. With probability $1-p$, however, at least one
of the parties detects that the transcript is not consistent
with his or her input and outputs a random bit. By the definition
of the problem, in this case the success probability is $\frac{1}{2}$.
Therefore the overall success probability of $P'$ is at least
$$ \frac{1}{2}+\frac{1}{2\sqrt{2}} - \frac{1-p}{2},$$
which must be at most $\frac{3}{4}$ by the CHSH inequality.
\end{proof}

\subsection*{Acknowledgements}

Part of this work was done while the authors were visiting the Institut Henri Poincar{\'e}
as part of the program ``Quantum information, computation and complexity'',
and we would like to thank the organizers for their efforts.  We thank Aram Harrow for discussions about lower bounds, Falk Unger for assistance with the proof of Lemma~\ref{lemma:1}, and Peter Harremo{\"e}s for discussions about Schoenberg's theorem.

\appendix


\section{A classical reformulation}
\label{apndx:classical}

The equivalence between Problem~\ref{prob:simulating} and
Problem~\ref{prob:simulatingclassical} is due to Tsirelson~\cite{Tsirelson:85b}. Here we sketch
only the easy direction of this equivalence: a solution
to Problem~\ref{prob:simulatingclassical} implies a solution
to Problem~\ref{prob:simulating}.

Let $\rho$ be a state on $\CC^d \otimes \CC^d$, and
let $A$ and $B$ be $d \times d$ Hermitian matrices whose eigenvalues
are in $\{-1,1\}$. The goal is for Alice and Bob to output
bits $\alpha$ and $\beta$ such that
$$
\Exp[\alpha \beta] = \tr \left( A \otimes B \cdot \rho \right).
$$

Let $a_1, \ldots, a_{d^2} \in \CC$ be the $d^2$ entries of the matrix $A \otimes \one_B \sqrt{\rho} $,
and similarly let $b_1, \ldots, b_{d^2} \in \CC$ be those of $\one_A \otimes B \sqrt{\rho}$.
Let $n = 2d^2$ and define the $n$-dimensional real vectors
\begin{align*}
{\vec a} &= (\Ree a_1, \ldots, \Ree a_{d^2}, \Imm a_1, \ldots,  \Imm a_{d^2}),\\
{\vec b} &= (\Ree b_1, \ldots, \Ree b_{d^2}, \Imm b_1, \ldots,  \Imm b_{d^2}).
\end{align*}
Then
$$\ip{\vec a, \vec a}  = \sum_{j=1}^{d^2} |a_j|^2 = \tr (A\otimes \one_B \rho A \otimes \one_B) = \tr (A^2 \otimes \one_B \cdot \rho) = \tr (\rho) = 1$$
and similarly $\ip{\vec b, \vec b} = 1$.
Moreover,
$$\ip{\vec a, \vec b} = \sum_{j=1}^{d^2} a_j \cdot b_j^* = \tr (A\otimes \one_B \cdot \rho \cdot \one_A \otimes B) =
   \tr  \left( { A}\otimes{B}\cdot \rho \right).$$
Hence, Alice and Bob can use $\vec a$ and $\vec b$ as input to
Problem~\ref{prob:simulatingclassical} in order to solve Problem~\ref{prob:simulating}.

\end{document}